\documentclass[11pt]{article}
\usepackage{amsmath,amssymb,amsthm,geometry,bm}
\usepackage{setspace}
\usepackage[authoryear,longnamesfirst,round]{natbib}
\usepackage[colorlinks=true, linkcolor=blue, citecolor=blue, urlcolor=blue]{hyperref}
\usepackage{etoc}

\newcommand{\R}{\mathbb{R}}
\newcommand{\E}{\mathbb{E}}

\newcommand{\one}{\mathbf{1}}
\DeclareMathOperator*{\argmin}{argmin}

\theoremstyle{plain}
\newtheorem{theorem}{Theorem}

\newtheorem{lemma}[theorem]{Lemma}
\newtheorem{corollary}[theorem]{Corollary}

\theoremstyle{definition}

\newtheorem{assumption}{Assumption}

\theoremstyle{remark}
\newtheorem{remark}[theorem]{Remark}

\title{Factor-Augmented Machine Learning Panel Regressions}
\author{Andrii Babii\\
		Department of Economics, University of North Carolina -- Chapel Hill\\
		and \\
		Luca Barbaglia \\
		European Commission, Joint Research Centre (JRC)\\
		and \\
		Eric Ghysels\\
		Department of Economics and Kenan-Flagler Business School, \\University of North Carolina -- Chapel Hill\\
		and \\
		Jonas Striaukas\\
		Department of Finance, Copenhagen Business School
	}
\date{\today}

\begin{document}
\maketitle

\begin{abstract}
	This paper develops the asymptotic theory for high-dimensional panel data regressions in settings with cross-sectionally dependent errors driven by common shocks. We consider a factor-augmented sparse-group LASSO estimator that combines MIDAS aggregation with latent factors. The estimator can take advantage of the mixed-frequency group structure in the time-series dimension. Theory shows that it can outperform the standard LASSO estimator both for prediction and estimation while allowing for cross-sectional dependence.
\end{abstract}



\etocsettocdepth.toc{none}

\section{Introduction}
Modern nowcasting and forecasting problems often rely on high-dimensional panel datasets with many cross-sectional units, a large number of predictors, and mixed-frequency data arriving in real time. Such panels typically exhibit strong cross-sectional dependence, reflecting common shocks in macroeconomic and financial time series. At the same time, the predictive component is often both sparse and dense: a few observed predictors may be important, after controlling for common factors generating broad comovement across units and variables.

This paper develops asymptotic theory for a factor-augmented sparse-group LASSO framework for nowcasting and forecasting with mixed-frequency panels.
The previous literature on mixed-frequency regressions includes the low-dimensional MIDAS regressions introduced in \cite{ghysels2006predicting}, \cite{ghysels2007midas}, and \cite{andreou2010regression}.

There is also a large literature on high-dimensional machine learning methods for i.i.d. data, e.g. \cite{bickel2009simultaneous}, \cite{belloni2014inference}, and \cite{cai2022sparse}.

\cite{babii2022machine} proposed to use the sparse-group LASSO of \cite{simon2013sparse} for high-dimensional time series and developed the asymptotic theory. \cite{mogliani2021bayesian} proposed a Bayesian version of the group LASSO for MIDAS regressions. Existing high-dimensional panel regression methods, such as \citet{babii2022machinepanel}, allow for high-dimensional covariates under approximate sparsity but do not exploit factor structure to capture cross-sectional dependence in the errors. Related sparse-plus-dense and factor-based panel approaches, see \cite{hansen2016factor}, \cite{ruecker2022estimation}, and \cite{fan2022bridging}, do not cover the mixed-frequency nowcasting problem with MIDAS aggregation and sparse-group regularization. The closest work to this paper is \cite{beyhum2024factor}, which studies factor-augmented sparse MIDAS regressions for time series, but not panel data regressions.

\paragraph*{Notation}

For an integer $N\in\mathbb{N}$, let $[N]=\{1,\ldots,N\}$. For a finite set $C$, let $|C|$ denote its cardinality. For a vector $v\in\R^p$, let $|v|_q=(\sum_{j=1}^p|v_j|^q)^{1/q}$ denote the $\ell_q$-norm with $|v|_\infty = \max_{j\in[p]}|v_j|$ and $|v|_0=\sum_{j=1}^p\one_{v_j\ne 0}$. A group is defined as a set of indices $G\subset[p]$. Let $v_G\in\R^p$ be a vector such that $(v_G)_i=v_i$ if $i\in G$ and $(v_G)_i=0$ otherwise. For vectors $v,w\in\R^n$, let $\langle v,w\rangle_n=n^{-1}\langle v,w\rangle$, where $\langle\cdot,\cdot\rangle$ denotes the Euclidean inner product, and let $\|\cdot\|_n=\sqrt{\langle\cdot,\cdot\rangle_n}$ be the associated norm. For a matrix $A\in\R^{n \times m}$, let $\|A\|_{\rm op}=\sup_{|v|_2=1}|Av|_2$. For a collection of groups $\mathcal{G}=\{G_1,\dots,G_K\}$, let $\|v\|_{2,1}=\sum_{k=1}^K|v_{G_k}|_2$ and $v_{\mathcal{G}}$ be the vector such that $(v_{\mathcal{G}})_i=v_i$ if $i\in\mathcal{G}$ and $(v_{\mathcal{G}})_i=0$ otherwise. The quantity $n_1\vee n_2$ is the maximum of $n_1$ and $n_2$. For a matrix $A$, let $\sigma_r(A),\sigma_{\min}(A)$ and $\sigma_{\max}(A)$ denote the $r^{\rm th}$ largest, smallest, and largest singular values of $A$, respectively.

\section{Factor-Augmented Machine Learning Panel Regressions}\label{sec:factor_aug_panel}
Consider the following factor-augmented predictive regression:
\begin{equation}\label{eq:panel_factor_reg}
	y_{i,t+1} = q_{i,t}^\top\delta + e_{i,t+1},\qquad e_{i,t+1} = f_{i,t}^\top\gamma + a_{i,t}+\varepsilon_{i,t+1},\qquad i\in[N],\;t\in[T],
\end{equation}
where $q_{i,t}\in\R^p$ includes all observed covariates, $f_{i,t}\in\R^R$ are latent factors, and $a_{i,t}\in\R$ is a misspecification error due to approximate sparsity or misspecified factor structure. The latent factors $f_{i,t}$ and the misspecification error $a_{i,t}$ can generate the cross-sectional dependence in errors.\footnote{This covers, for example, the interactive-effects specification $f_{i,t}=\Lambda_i g_t$, where $g_t\in\R^r$ is a common factor and $\Lambda_i\in\R^{R\times r}$ is a unit-specific loading matrix. Cross-sectional dependence then arises through the common factor; see also \cite{babii2025tensor} for a more refined tensor factor model analysis.} Therefore, the model represents a flexible specification with approximately sparse signals, dense components, and cross-sectional dependence.

\subsection{MIDAS}
Covariates may be measured at higher frequencies than the outcome $y_{i,t+1}$. In this case we will transform the high-frequency data into Equation~\eqref{eq:panel_factor_reg} using the MIDAS approach of \cite{babii2022machine}. Consider a panel of $K_x$ covariates for $N$ units observed over $T$ low-frequency periods, with $m_x$ high-frequency observations for each low-frequency period:\footnote{For simplicity of presentation, we assume that the number of high-frequency observations is the same across all low-frequency periods. This assumption can be easily relaxed, and the observations can also be allowed to extend forward into the nowcasting period.}
\[
	\left\{x_{i,t-(j-1)/m_x,k}^H:\ i\in[N],\ t\in[T],\ j\in[m_x],\ k\in[K_x]\right\}.
\]
Suppose this panel follows the approximate factor model:
\begin{equation}\label{eq:panel_factor_x}
	x_{i,t-(j-1)/m_x,k}^H = b_k^\top f_{i,t-(j-1)/m_x}^H + u_{i,t-(j-1)/m_x,k}^H,
\end{equation}
where $b_k\in\R^{K_f}$ is the loading vector, $f_{i,t-(j-1)/m_x}^H\in\R^{K_f}$ is the high-frequency factor, and $u_{i,t-(j-1)/m_x,k}^H$ is the idiosyncratic component. A second panel of $K_z$ high-frequency covariates
\[
	\left\{z_{i,t-(j-1)/m_z,k}^H:\ i\in[N],\ t\in[T],\ j\in[m_z],\ k\in[K_z]\right\}
\]
need not follow a factor model and will enter Equation~\eqref{eq:panel_factor_reg} only through the approximately sparse coefficients.

For each unit $i$ and covariate $k$, collect the high-frequency observations in matrices $X_{i,k}^H :=(x_{i,t-(j-1)/m_x,k}^H)_{t\in[T],\,j\in[m_x]}$ and $Z_{i,k}^H:=(z_{i,t-(j-1)/m_z,k}^H)_{t\in[T],\,j\in[m_z]}$.
These high-frequency observations can be mapped with MIDAS weights into the regression model in Equation~\eqref{eq:panel_factor_reg} as follows. Let $w_l:[0,1]\to\R$ for $l\in[L]$ be a dictionary of functions used to approximate the MIDAS weights. Define the weighting matrices $W^x:=\left(w_l((j-1)/m_x)/m_x\right)_{j\in[m_x],l\in[L]}$ and $W^z:=\left(w_l((j-1)/m_z)/m_z\right)_{j\in[m_z],l\in[L]}.$
Then $X_{i,k}^H W^x\in\R^{T\times L}$ and $Z_{i,k}^H W^z\in\R^{T\times L}$ are the MIDAS-weighted versions of high-frequency observations. Define
\begin{equation}\label{eq:panel_midas_x}
	\mathbf{x}_i:=(X_{i,1}^H W^x,\dots,X_{i,K_x}^H W^x)\in\R^{T\times p_x},
	\qquad
	\mathbf{z}_i:=(Z_{i,1}^H W^z,\dots,Z_{i,K_z}^H W^z)\in\R^{T\times p_z},
\end{equation}
where $p_x=LK_x$ and $p_z=LK_z$. Let $x_{i,t}\in\R^{p_x}$ and $z_{i,t}\in\R^{p_z}$ be the $t^{\text{th}}$ rows of $\mathbf{x}_i$ and $\mathbf{z}_i$ transposed, respectively. The aggregated covariates are collected in a vector $q_{i,t}:=(1,z_{i,t}^\top,x_{i,t}^\top)^\top\in\R^p$, where $p=p_z+p_x+1$, and we assume that the initial observations of lagged variables are available. Then $q_{i,t}$ corresponds to the vector of covariates in Equation~\eqref{eq:panel_factor_reg}.

\subsection{Factor-Augmented sg-LASSO}
For estimation purposes, using matrix notation, define $\mathbf{q}_i:=(q_{i,1},\dots,q_{i,T})^\top\in\R^{T\times p}$, $\mathbf{y}_i:=(y_{i,2},\dots,y_{i,T+1})^\top\in\R^{T}$, and stack $\mathbf{y}:=(\mathbf{y}_1^\top,\dots,\mathbf{y}_N^\top)^\top\in\R^{NT}$, $\mathbf{Q}:=(\mathbf{q}_1^\top,\dots,\mathbf{q}_N^\top)^\top\in\R^{NT\times p}$, and $\mathbf{X}:=(\mathbf{x}_1^\top,\dots,\mathbf{x}_N^\top)^\top\in\R^{NT\times p_x}$. Let also $\mathbf{A}\in\R^{NT}$ and $\mathbf{E}\in\R^{NT}$ be the stacked versions of $a_{i,t}$ and $\varepsilon_{i,t+1}$. Using matrix notation, Equation~\eqref{eq:panel_factor_reg} becomes
\begin{equation*}
	\mathbf{y}=\mathbf{Q}\delta+\mathbf{F}\gamma+\mathbf{A}+\mathbf E.
\end{equation*}

The factor structure of covariates, see Equation~\eqref{eq:panel_factor_x}, implies that the dense block of covariates can be written as
\begin{equation}\label{eq:factor_model}
	\mathbf{X}=\mathbf{F}B^\top + \mathbf{U},\qquad \E[\mathbf{U}|\mathbf{F},B]=0,
\end{equation}
where $\mathbf{F}\in\R^{NT\times R}$ are the (MIDAS-aggregated) factors with $R=LK_f$, $B\in\R^{p_x\times R}$ are the factor loadings, and $\mathbf{U}\in\R^{NT\times p_x}$ are the idiosyncratic components. The MIDAS aggregation introduces a group structure in the regression model, where each original predictor contributes a group of $L$ dictionary coefficients; see \cite{babii2022machine}. The group structure with $K$ groups is described as a partition $\{G_1,\dots,G_K\}$ of $[p]=\{1,\dots,p\}$.

We extract factors from $\mathbf{X}$ using PCA. Let the columns of $\hat{\mathbf{F}}/\sqrt{NT}$ be the eigenvectors associated with the largest $R$ eigenvalues of $\mathbf{X}\mathbf{X}^\top$, so that $\frac{1}{NT}\hat{\mathbf{F}}^\top\hat{\mathbf{F}}=I_R$. The factor-augmented sg-LASSO estimator is
\begin{equation}\label{eq:fapd_est}
	(\hat\delta,\hat\gamma)\in\argmin_{(d,c)\in\R^p\times\R^R}
	\|\mathbf{y}-\mathbf{Q}d-\hat{\mathbf{F}}c\|_{NT}^2
	+2\lambda_1|d|_1+2\lambda_2\|d\|_{2,1},
\end{equation}
where $|d|_1 = \sum_{j=1}^p |d_j|$ and $\|d\|_{2,1} = \sum_{k=1}^K |d_{G_k}|_2$ are the $\ell_1$ and $\ell_{2,1}$ norms, respectively, and $\lambda_1,\lambda_2\geq0$ are tuning parameters. The factor-augmented sg-LASSO thus captures both the dense component of the covariates through the factors and the approximately sparse component through the penalized regression on $\mathbf{Q}$. The $\ell_1$ LASSO penalty encourages sparsity at the coordinate level, while the $\ell_{2,1}$ group LASSO penalty encourages sparsity at the group level.

The estimator in Equation~\eqref{eq:fapd_est} combines a dense component, represented by the factors estimated from $\mathbf{X}$, with a sparse component, represented by the penalized regression on $\mathbf{Q}$; see also \cite{beyhum2024factor}. We consider an asymptotic regime in which $N,T\to\infty$, the number of covariates $p\to\infty$ may grow with $N$ and $T$, and the number of factors $R$ remains fixed. The fixed-$R$ assumption can also be relaxed; see \cite{babii2025tensor}, \cite{beyhum2022factor}, and \cite{freeman2023linear}.

Let $S=\{j\in[p]:\delta_j\ne0\}$ be the active coordinates in the true parameter $\delta\in\R^p$ and let $\mathcal A=\{k\in[K]:\delta_{G_k}\ne0\}$ be the active group indices. Put $s=|S|$, $m=|\mathcal A|$, and for an integer $\ell\geq1$, let
\begin{equation*}
	d_{(\ell)}=\max_{\substack{I\subset[K]:|I|=\ell}}\sum_{k\in I}|G_k|
\end{equation*}
be the total number of elements in the largest $\ell$ groups. Let $P_{\hat{\mathbf{F}}}=(NT)^{-1}\hat{\mathbf{F}}\hat{\mathbf{F}}^\top$, $M_{\hat{\mathbf{F}}}=I_{NT}-P_{\hat{\mathbf{F}}}$, and $\tilde{\mathbf{Q}}=M_{\hat{\mathbf{F}}}\mathbf{Q}$.

\begin{assumption}\label{as:fapd_hl_score}
	$\exists\sigma>0$ such that $\E\left[e^{u a^\top\tilde{\mathbf{Q}}^\top\mathbf{E}}|\mathbf{Q}\right] \leq \exp\left(\frac{u^2\sigma^2NT|a|_2^2}{2}\right),\forall u\in\R,a\in\R^p$.
\end{assumption}

\begin{assumption}\label{as:fapd_hl_re}
	$\exists\kappa>0$ such that $\|\tilde{\mathbf Q}v\|_{NT}^2 \geq \kappa^2|v|_2^2,\forall v\in\R^p$ such that $|v|_1/\sqrt{s}+\|v\|_{2,1}/\sqrt{m}\leq 10|v|_2$ with probability approaching one.
\end{assumption}

\begin{assumption}\label{as:tuning}
	The tuning parameters are such that $\lambda_1=r_{s,m}/\sqrt{s}$ and $\lambda_2=r_{s,m}/\sqrt{m}$ with
	\begin{equation*}
		r_{s,m} = 4\sqrt{2}\sigma\sqrt{\frac{m\log(eK/m) + s\log(5ed_{(m)}/s) + \log(2/\epsilon)}{NT}},\qquad \epsilon\in(0,1).
	\end{equation*}
\end{assumption}
Assumption~\ref{as:fapd_hl_score} requires that regression errors are sub-Gaussian conditionally on the covariates. It can be relaxed in several different ways, e.g. winsorizing the data or robustifying the least-squares objective function with Huber or median-of-means losses; see \cite{lugosi2019mean} for more details. Assumption~\ref{as:fapd_hl_re} is a restricted eigenvalue condition and allows for singularity of the design matrix. Assumption~\ref{as:tuning} requires the oracle rates for tuning parameters. It is also known that the LASSO estimator with tuning parameters selected by cross-validation achieves similar results as the oracle LASSO estimator; see \cite{chetverikov2021cross}.

The following result holds:
\begin{theorem}\label{thm:fapd_oracle}
	Suppose that Assumptions~\ref{as:fapd_hl_score}, \ref{as:fapd_hl_re}, and \ref{as:tuning} hold. Then with probability $1-\epsilon-o(1)$,
	\begin{equation*}
		\left\|\mathbf{Q}\hat\delta+\hat{\mathbf{F}}\hat\gamma-\mathbf{Q}\delta-\mathbf{F}\gamma\right\|_{NT} \leq \frac{3.5r_{s,m}}{\kappa} + \sqrt{2}\|\mathbf{A}\|_{NT} + 2\|M_{\hat{\mathbf{F}}}\mathbf{F}\gamma\|_{NT} + \|P_{\hat{\mathbf{F}}}\mathbf E\|_{NT}
	\end{equation*}
	and
	\begin{equation*}
		s^{-1/2}|\hat\delta - \delta|_1 + m^{-1/2}\|\hat\delta - \delta\|_{2,1} \leq \frac{40r_{s,m}}{\kappa^2} + \frac{5}{r_{s,m}}\|M_{\hat{\mathbf{F}}}\mathbf{F}\gamma + M_{\hat{\mathbf{F}}}\mathbf{A}\|_{NT}^2.
	\end{equation*}
\end{theorem}
Theorem~\ref{thm:fapd_oracle} states the oracle inequality for the factor-augmented sg-LASSO estimator. The proof of this and other results can be found in the Appendix. Next, we deduce the convergence rate of the factor-augmented sg-LASSO estimator under additional assumptions. Let $u_{i,t}\in\R^{p_x}$ be a column vector corresponding to the row for observation $(i,t)$ in $\mathbf{U}$.
\begin{assumption}\label{as:fapd_factor}
	The following conditions hold:
	\begin{itemize}
		\item[(i)] $\frac{1}{NT}\mathbf{F}^\top \mathbf{F}\xrightarrow{p}I_R$ and $\frac{1}{p_x}B^\top B\xrightarrow{p}D$, where $D$ is a diagonal, positive definite matrix;
		\item[(ii)] $\exists C<\infty$ such that $\sup_{i,t}\lambda_{\max}\left(\E\left(u_{i,t}u_{i,t}^\top \mid B\right)\right)\leq C$;
		\item[(iii)] $|\gamma|_2=O(1)$;
		\item[(iv)] $\exists\sigma>0$ with $\E[e^{ua^\top \mathbf{E}}|\mathbf{X}]\leq \exp\left(\sigma^2u^2|a|^2_2/2\right),\forall u\in\R,a\in\R^{NT}$;
		\item[(v)] $\exists C<\infty$ such that $\|\mathrm{Var}(\mathbf E\mid \mathbf{X})\|_{\rm op}\leq C$ and the entries of $\mathbf{E}$ are independent given $\mathbf{X}$;
		\item[(vi)] $\max_{i,s}\E|u_{i,s}|_2^4=O(p_x^2)$ and $\max_{(i,s)\ne(j,t)}\E(u_{i,s}^\top u_{j,t})^2 = O(p_x)$.
	\end{itemize}
\end{assumption}
Assumption~\ref{as:fapd_factor} imposes a set of restrictions needed to characterize the estimation error in factors; see \cite{bai2003inferential}. Some of the more high-level conditions in this assumption can also be reduced to a set of more primitive conditions; see e.g. \cite{beyhum2024factor} and references therein.

\begin{corollary}\label{cor:fapd_rate}
	Suppose that Assumptions~\ref{as:fapd_hl_score}, \ref{as:fapd_hl_re}, \ref{as:tuning}, and \ref{as:fapd_factor} hold. Then
	\begin{equation*}
		\|\mathbf{Q}\hat\delta+\hat{\mathbf{F}}\hat\gamma-\mathbf{Q}\delta-\mathbf{F}\gamma\|_{NT}^2 = O_P\left(\frac{m\log(eK/m)+s\log(ed_{(m)}/s)}{NT}\right)
	\end{equation*}
	and
	\begin{equation*}
		|\hat\delta - \delta|_1 = O_P\left(\sqrt{\frac{sm\log(eK/m)+s^2\log(ed_{(m)}/s)}{NT}}\right),
	\end{equation*}
	provided that $\|\mathbf A\|_{NT} = O(r_{s,m})$ and $p_x^{-1} = O(r_{s,m}^2)$.
\end{corollary}

\begin{remark}
	Theorem~\ref{thm:fapd_oracle} and Corollary~\ref{cor:fapd_rate} extend the time-series framework of \cite{beyhum2024factor} to panel data. That framework extends \cite{babii2022machine}. These papers derive the rate for time-series data when $N=1$, which under comparable assumptions has order $O_P\left(\frac{(s\vee m)\log p}{T}\right)$. For panel data, \cite{babii2022machinepanel} derive the rate of order $O_P\left(\frac{(s\vee m)\log p}{NT}\right)$. Neither of these rates reflects the advantage of the sparse-group structure of the regression model, and both rates are slower than the rate in Corollary~\ref{cor:fapd_rate}.
\end{remark}
\begin{remark}
	For cross-sectional data when $T=1$, Theorem 5 of \cite{cai2022sparse} derives the minimax lower bound on the convergence rate in $\ell_2$ norm for any estimator of order
	\[
		O_P\left(\frac{m\log(eK/m) + s\log(em G_*/s)}{N}\right),
	\]
	provided that all groups have the same size $G_*$. In this case, the size of $m$ largest groups becomes $d_{(m)}=mG_*$. This rate can be faster than the optimal $O_P(s\log(p/s)/N)$ rate of the LASSO estimator obtained in \cite{bellec2018slope}.
\end{remark}

\newpage
\appendix
\begin{center}
	\Large\bfseries APPENDIX
\end{center}

%\tableofcontents
%\etocsettocdepth.toc{subsection}

\setcounter{theorem}{0}
\setcounter{assumption}{0}
\renewcommand{\thetheorem}{\thesection.\arabic{theorem}}
\renewcommand{\theassumption}{\thesection.\arabic{assumption}}

\section{Proofs}
\subsection{Proof of Theorem~\ref{thm:fapd_oracle}}
\begin{lemma}\label{lem:fapd_score_bound}
	Under Assumptions~\ref{as:fapd_hl_score} and \ref{as:tuning}, we have with probability at least $1-\epsilon$,
	\begin{equation*}
		\left|\left\langle\tilde{\mathbf{Q}}^\top\mathbf E,v\right\rangle_{NT}\right|
		\leq \frac{r_{s,m}}{2}
		\left(|v|_2+\frac{|v|_1}{\sqrt{s}}+\frac{\|v\|_{2,1}}{\sqrt{m}}\right), \qquad \forall v\in\R^p.
	\end{equation*}
\end{lemma}
\begin{proof}[Proof of Lemma~\ref{lem:fapd_score_bound}]
	Let $\xi:=\frac{\tilde{\mathbf Q}^{\top}\mathbf E}{NT}$ denote the score. Let $\{G_1,\dots,G_K\}$ denote the groups corresponding to the population parameter $\delta\in\R^p$. For $v\in\mathbb R^p$, let ${\mathcal A}(v):=\{k\in[K]:v_{G_k}\neq0\}$ be the groups touched by $v$. Define the following norm of $\xi$:
	\begin{equation*}
		\rho_{s,m}(\xi) := \sup\left\{|v^\top \xi|: v\in\mathcal{V}_{s,m}\right\},\qquad \mathcal{V}_{s,m} := \{v\in\R^p:|v|_2=1, |\mathcal{A}(v)|\leq m, |v|_0\leq s\}.
	\end{equation*}
	We will show first the following deterministic inequality
	\begin{equation}\label{eq:score_bound}
		|v^\top \xi| \leq \rho_{s,m}(\xi) \left(|v|_2+\frac{|v|_1}{\sqrt{s}}+\frac{\|v\|_{2,1}}{\sqrt{m}} \right), \qquad \forall v\in\mathbb R^p.
	\end{equation}
	Fix $v\in\mathbb R^p$ and order the groups as follows:
	\begin{equation*}
		|v_{G_{(1)}}|_2\geq |v_{G_{(2)}}|_2\geq \cdots \geq |v_{G_{(K)}}|_2.
	\end{equation*}
	Split the ordered groups into consecutive blocks of size $m$ with $\ell = 0,1,\ldots$
    \begin{equation*}
        B_\ell := \left\{k\in[K]: \ell m< k \leq \min\{(\ell+1)m,K\}\right\},
        \qquad \mathcal{G}_\ell := \bigcup_{k\in B_\ell}G_{(k)}\subseteq[p].
    \end{equation*}
	Then
    \begin{equation}\label{eq:group_block}
        \sum_{\ell\geq0}|v_{\mathcal{G}_\ell}|_2 = |v_{\mathcal{G}_0}|_2+\sum_{\ell\geq1}|v_{\mathcal{G}_\ell}|_2 \leq |v|_2+\frac{\|v\|_{2,1}}{\sqrt m},
    \end{equation}
	where the inequality follows from the fact that the groups are ordered by $\ell_2$ norm:
	\begin{equation*}
		|v_{\mathcal{G}_\ell}|_2\leq \sqrt{m}\max_{k\in B_{\ell}}|v_{G_{(k)}}|_2 \leq \sqrt{m}\min_{k\in B_{\ell-1}}|v_{G_{(k)}}|_2 \leq \frac{1}{\sqrt m}\sum_{k\in B_{\ell-1}}|v_{G_{(k)}}|_2.
	\end{equation*}

    Now fix $\ell$, order the coordinates in $\mathcal{G}_\ell$ by decreasing absolute value of $v_j$, and split them into consecutive coordinate blocks
    \begin{equation*}
        J_{\ell,0},J_{\ell,1},J_{\ell,2},\ldots,
    \end{equation*}
    each of cardinality at most $s$. Then each $J_{\ell,q}$ is contained in at most $m$ groups and has at most $s$ coordinates. By a similar argument as above,
    \begin{equation*}
        \sum_{q\geq0}|v_{J_{\ell,q}}|_2
        \leq
        |v_{\mathcal{G}_\ell}|_2+\frac{|v_{\mathcal{G}_\ell}|_1}{\sqrt s}.
    \end{equation*}
    Summing over $\ell$ and using the bound in Equation~\eqref{eq:group_block} gives
    \begin{equation*}
        \sum_{\ell\geq0}\sum_{q\geq0}|v_{J_{\ell,q}}|_2
        \leq
        |v|_2+\frac{\|v\|_{2,1}}{\sqrt m}+\frac{|v|_1}{\sqrt s}.
    \end{equation*}
    Since each $v_{J_{\ell,q}}$ is supported on at most $s$ coordinates and at most $m$ groups, we have
    \begin{equation*}
        |\xi^\top v| \leq \sum_{\ell\geq0}\sum_{q\geq0}|\xi^\top v_{J_{\ell,q}}|
        \leq \rho_{s,m}(\xi) \sum_{\ell\geq0}\sum_{q\geq0}|v_{J_{\ell,q}}|_2 \leq \rho_{s,m}(\xi) \left(|v|_2+\frac{|v|_1}{\sqrt{s}}+\frac{\|v\|_{2,1}}{\sqrt{m}}\right),
    \end{equation*}
    which proves the inequality in Equation~\eqref{eq:score_bound}.

	It remains to show that $\rho_{s,m}(\xi) \leq r_{s,m}/2$ with probability at least $1-\epsilon$. To that end, fix $B\subset[K]$ with $|B|=m$, and write $\mathcal{G}_B:=\bigcup_{g\in B}G_g$. By definition of $d_{(m)}$, we have
	\begin{equation*}
		|\mathcal{G}_B|\leq \max_{B\subset[K]:|B|=m}\left|\bigcup_{k\in B}G_k\right| = d_{(m)}.
	\end{equation*}
	For every $J\subset \mathcal{G}_B$ with $|J|\leq s$, let $\mathcal N_J$ be a $1/2$-net of $\{v\in\R^p:|v|_2=1, \mathrm{supp}(v)\subset J\}$ with respect to $\ell_2$ norm. We can choose $\mathcal N_J$ so that its covering number is $|\mathcal N_J|\leq 5^s$; see \cite{Vershynin_2026}, Corollary 4.2.11. Let
    \begin{equation*}
        \mathcal N:=\bigcup_{\substack{B\subset[K]\\ |B|=m}} \bigcup_{\substack{J\subset \mathcal{G}_B\\ |J|\leq s}}\mathcal N_J.
    \end{equation*}
    The cardinality of this net satisfies
    \begin{equation*}
        |\mathcal N|\leq \binom{K}{m}\max_{B\subset[K]:|B|=m}\sum_{j=1}^s\binom{|\mathcal{G}_B|}{j}5^s \leq \left(\frac{eK}{m}\right)^m\left(\frac{5ed_{(m)}}{s}\right)^s,
    \end{equation*}
	where we use $\binom{n}{k}\leq \sum_{j=0}^k\binom{n}{j} \leq \left(\frac{en}{k}\right)^k$ since $s\leq d_{(m)}$; see \cite{Vershynin_2026}, Exercise 0.6.

	For every $v\in\R^p$ satisfying $|v|_2=1$, $|v|_0\leq s$, and $|\mathcal A(v)|\leq m$, its support is contained in $\mathcal{G}_B$ for some $B\subset[K]$ with $|B|=m$. Let $J:=\operatorname{supp}(v)$, so that $J\subset \mathcal G_B$ and $|J|\leq s$. Since $\mathcal N_J$ is a $1/2$-net, there exists $b\in\mathcal N_J\subset\mathcal N$ such that $|v-b|_2\leq 1/2$. Moreover, $v-b$ is supported on $J$, and therefore it is supported on at most $s$ coordinates and touches at most $m$ groups. Hence,
	\begin{equation*}
		|v^\top\xi| \leq |b^\top \xi| + |(v-b)^\top\xi| \leq \max_{c\in\mathcal N}|\xi^\top c| + \rho_{s,m}(\xi)|v-b|_2 \leq \max_{c\in\mathcal N}|\xi^\top c| + \frac{1}{2}\rho_{s,m}(\xi).
	\end{equation*}
	Taking the supremum over all such $v$ gives
	\begin{equation*}
		\rho_{s,m}(\xi)\leq 2\max_{c\in\mathcal N}|\xi^\top c|.
	\end{equation*}
    Therefore, by the union bound, for any $u>0$
    \begin{equation*}
        \Pr(\rho_{s,m}(\xi)>2u)\leq \Pr\left(\max_{c\in\mathcal N}|c^\top\xi|>u\right)\leq |\mathcal{N}|\max_{c\in\mathcal{N}}\Pr\left(|c^\top\xi|>u\right),
    \end{equation*}
	where $|\mathcal N|\leq \left(\frac{eK}{m}\right)^m\left(\frac{5ed_{(m)}}{s}\right)^s$. Under Assumption~\ref{as:fapd_hl_score},
	\begin{equation*}
		\Pr\left(|c^\top\xi|>u\right) \leq 2\exp\left(-\frac{NTu^2}{2\sigma^2}\right),\qquad \forall c\in\mathcal{N}.
	\end{equation*}
	Choosing $u = \sigma\sqrt{\frac{2\log|\mathcal{N}| + 2\log(2/\epsilon)}{NT}}$, we obtain with probability at least $1-\epsilon$,
    \begin{equation*}
        \rho_{s,m}(\xi) \leq 2\sqrt{2}\sigma\sqrt{\frac{m\log(eK/m) + s\log(5ed_{(m)}/s) + \log(2/\epsilon)}{NT}} = \frac{r_{s,m}}{2}.
    \end{equation*}
\end{proof}

\begin{proof}[Proof of Theorem~\ref{thm:fapd_oracle}]
	By the definition of $M_{\hat{\mathbf{F}}}$, concentrating out the factors gives
	\begin{equation*}
		\hat\delta\in\arg\min_{d\in\R^p}
		\|\tilde{\mathbf{y}}-\tilde{\mathbf{Q}}d\|_{NT}^2 + 2\Omega(d),
	\end{equation*}
	where $\Omega(d):=\lambda_1|d|_1+\lambda_2\|d\|_{2,1}$, $\tilde{\mathbf{Q}}=M_{\hat{\mathbf{F}}}\mathbf{Q}$, and
	\begin{equation}\label{eq:fapd_y_tilde}
		\tilde{\mathbf{y}}=M_{\hat{\mathbf{F}}}\mathbf{y} = \tilde{\mathbf{Q}}\delta + \underbrace{M_{\hat{\mathbf{F}}}\mathbf{F}\gamma + M_{\hat{\mathbf{F}}}\mathbf{A}}_{=:\eta} + M_{\hat{\mathbf{F}}}\mathbf E.
	\end{equation}
	By Fermat's rule, $\hat\delta$ satisfies
	\begin{equation*}
		\frac{1}{NT}\tilde{\mathbf{Q}}^\top(\tilde{\mathbf{Q}}\hat\delta-\tilde{\mathbf{y}}) + z^*=0,\qquad z^*\in\partial\Omega(\hat\delta),
	\end{equation*}
	where $\partial\Omega(\hat\delta)$ is the subdifferential of $\Omega$ at $\hat\delta$.
	The inner product with $\delta-\hat\delta$ and convexity of $d\mapsto\Omega(d)$ give
	\begin{equation}\label{eq:fapd_kkt}
		\left\langle \tilde{\mathbf{y}}-\tilde{\mathbf{Q}}\hat\delta,\tilde{\mathbf{Q}}(\delta-\hat\delta)\right\rangle_{NT} = \left\langle z^*,\delta-\hat\delta\right\rangle \leq \Omega(\delta)-\Omega(\hat\delta),
	\end{equation}
	Setting $\Delta:=\hat\delta-\delta$ and substituting the expression for $\tilde{\mathbf{y}}$ in Equation~\eqref{eq:fapd_y_tilde}, we obtain
	\begin{equation*}
		\begin{aligned}
			\|\tilde{\mathbf{Q}}\Delta\|_{NT}^2 & \leq \left\langle \mathbf E,\tilde{\mathbf{Q}}\Delta\right\rangle_{NT} + \left\langle\eta,\tilde{\mathbf{Q}}\Delta\right\rangle_{NT} + \Omega(\delta)-\Omega(\hat\delta) \\
			& \leq \left\langle\mathbf E,\tilde{\mathbf{Q}}\Delta\right\rangle_{NT} + \left\langle \eta,\tilde{\mathbf{Q}}\Delta\right\rangle_{NT} + \lambda_1\{|\Delta_S|_1-|\Delta_{S^c}|_1\} + \lambda_2\{\|\Delta_{\mathcal A}\|_{2,1} - \|\Delta_{\mathcal A^c}\|_{2,1}\},
		\end{aligned}
	\end{equation*}
	where the second line follows since $\delta_{S^c}=0$ and $\delta_{\mathcal{A}^c}=0$, so by the reverse triangle inequality
	\begin{equation*}
		\begin{aligned}
			|\delta|_1-|\delta+\Delta|_1 & \leq
			|\Delta_S|_1-|\Delta_{S^c}|_1, \\
			\|\delta\|_{2,1}-\|\delta+\Delta\|_{2,1}
			& \leq \|\Delta_{\mathcal A}\|_{2,1} - \|\Delta_{\mathcal A^c}\|_{2,1}.
		\end{aligned}
	\end{equation*}
	Moreover, under Assumption~\ref{as:fapd_hl_score} by Lemma~\ref{lem:fapd_score_bound}, with probability at least $1-\epsilon$,
	\begin{equation*}
		\left\langle\mathbf E,\tilde{\mathbf{Q}}v\right\rangle_{NT}
		\leq \frac{r_{s,m}}{2}
		\left(|v|_2+\frac{|v|_1}{\sqrt{s}}+\frac{\|v\|_{2,1}}{\sqrt{m}}\right) \qquad \text{for all }v\in\R^p.
	\end{equation*}
	Using $\lambda_1=\frac{r_{s,m}}{\sqrt{s}}$ and $\lambda_2=\frac{r_{s,m}}{\sqrt{m}}$, under Assumption~\ref{as:tuning}, and $|\Delta_S|_1\leq \sqrt{s}|\Delta|_2$, $\|\Delta_{\mathcal A}\|_{2,1}\leq \sqrt{m}|\Delta|_2$, we obtain
	\begin{equation}\label{eq:basic_inequality}
		\begin{aligned}
			\|\tilde{\mathbf Q}\Delta\|_{NT}^2 & \leq 3.5r_{s,m}|\Delta|_2 - \frac{r_{s,m}}{2}\left(\frac{|\Delta_{S^c}|_1}{\sqrt{s}} + \frac{\|\Delta_{\mathcal A^c}\|_{2,1}}{\sqrt{m}}\right) + \left\|\tilde{\mathbf{Q}}\Delta\right\|_{NT}\|\eta\|_{NT}.
		\end{aligned}
	\end{equation}

	\medskip

	\noindent \emph{Case 1:} Suppose that
	\begin{equation*}
		\frac{|\Delta|_1}{\sqrt{s}} + \frac{\|\Delta\|_{2,1}}{\sqrt{m}} \leq 10|\Delta|_2
	\end{equation*}
	Assumption~\ref{as:fapd_hl_re} then gives $\kappa|\Delta|_2\leq \|\tilde{\mathbf{Q}}\Delta\|_{NT}$. Dropping the term with negative sign in Equation~\eqref{eq:basic_inequality} gives
	\begin{equation*}
		\|\tilde{\mathbf{Q}}\Delta\|_{NT}^2 \leq \frac{3.5r_{s,m}}{\kappa}\|\tilde{\mathbf{Q}}\Delta\|_{NT} + \|\tilde{\mathbf{Q}}\Delta\|_{NT}\|\eta\|_{NT}.
	\end{equation*}
	Then
	\begin{equation*}
		\|\tilde{\mathbf{Q}}\Delta\|_{NT} \leq \frac{3.5r_{s,m}}{\kappa} + \|\eta\|_{NT}
	\end{equation*}
	and
	\begin{equation*}
		\frac{|\Delta|_1}{\sqrt{s}} + \frac{\|\Delta\|_{2,1}}{\sqrt{m}} \leq \frac{35r_{s,m}}{\kappa^2} + \frac{10}{\kappa}\|\eta\|_{NT} \leq \frac{40r_{s,m}}{\kappa^2} + \frac{5}{r_{s,m}}\|\eta\|_{NT}^2,
	\end{equation*}
	where the last inequality follows from $2ab\leq a^2+b^2$.

	\medskip

	\noindent \emph{Case 2:} Now suppose that
	\begin{equation*}
		\frac{|\Delta|_1}{\sqrt{s}} + \frac{\|\Delta\|_{2,1}}{\sqrt{m}} > 10|\Delta|_2.
	\end{equation*}
	Since
	\begin{equation*}
		\frac{|\Delta_S|_1}{\sqrt{s}} + \frac{\|\Delta_{\mathcal A}\|_{2,1}}{\sqrt{m}} \leq 2|\Delta|_2,
	\end{equation*}
	we have
	\begin{equation*}
		\begin{aligned}
			\|\tilde{\mathbf Q}\Delta\|_{NT}^2 & \leq 4.5r_{s,m}|\Delta|_2 - \frac{r_{s,m}}{2}\left(\frac{|\Delta|_1}{\sqrt{s}} + \frac{\|\Delta\|_{2,1}}{\sqrt{m}}\right) + \left\|\tilde{\mathbf{Q}}\Delta\right\|_{NT}\|\eta\|_{NT} \\
			& \leq -\frac{r_{s,m}}{20}\left(\frac{|\Delta|_1}{\sqrt{s}} + \frac{\|\Delta\|_{2,1}}{\sqrt{m}}\right) + \left\|\tilde{\mathbf{Q}}\Delta\right\|_{NT}\|\eta\|_{NT}.
		\end{aligned}
	\end{equation*}
	This shows that $\|\tilde{\mathbf{Q}}\Delta\|_{NT}\leq \|\eta\|_{NT}$ and that
	\begin{equation*}
		\frac{|\Delta|_1}{\sqrt{s}} + \frac{\|\Delta\|_{2,1}}{\sqrt{m}} \leq \frac{20}{r_{s,m}}\left(\left\|\tilde{\mathbf{Q}}\Delta\right\|_{NT}\|\eta\|_{NT} - \left\|\tilde{\mathbf{Q}}\Delta\right\|_{NT}^2\right) \leq \frac{5\|\eta\|_{NT}^2}{r_{s,m}}.
	\end{equation*}
	To sum up, in both cases, we have
	\begin{equation*}
		\|\tilde{\mathbf{Q}}\Delta\|_{NT} \leq \frac{3.5r_{s,m}}{\kappa} + \|\eta\|_{NT}
	\end{equation*}
	and
	\begin{equation*}
		\frac{|\Delta|_1}{\sqrt{s}} + \frac{\|\Delta\|_{2,1}}{\sqrt{m}} \leq \frac{40r_{s,m}}{\kappa^2} + \frac{5\|\eta\|_{NT}^2}{r_{s,m}}.
	\end{equation*}

	Next,
	\begin{equation*}
		\hat{\mathbf F}\hat\gamma = P_{\hat{\mathbf{F}}}(\mathbf y-\mathbf Q\hat\delta) = P_{\hat{\mathbf{F}}}\mathbf Q(\delta-\hat\delta) + P_{\hat{\mathbf{F}}}(\mathbf F\gamma + \mathbf A) + P_{\hat{\mathbf{F}}}\mathbf E,
	\end{equation*}
	so that
	\begin{equation*}
		\begin{aligned}
			\mathbf Q(\hat\delta-\delta) + \hat{\mathbf F}\hat\gamma - \mathbf F\gamma & = M_{\hat{\mathbf{F}}}\mathbf Q(\hat\delta-\delta) + M_{\hat{\mathbf{F}}}(\hat{\mathbf F}\hat \gamma - {\mathbf{F}}\gamma) + P_{\hat{\mathbf{F}}}\mathbf Q(\hat\delta-\delta) + P_{\hat{\mathbf{F}}}(\hat{\mathbf F}\hat \gamma - {\mathbf{F}}\gamma) \\
			& = \tilde {\mathbf Q}(\hat\delta-\delta) - M_{\hat{\mathbf{F}}}\mathbf F\gamma + P_{\hat{\mathbf{F}}}\mathbf Q(\hat\delta-\delta) + \hat{\mathbf F}\hat \gamma - P_{\hat{\mathbf{F}}}{\mathbf{F}}\gamma \\
			& = \tilde {\mathbf Q}(\hat\delta-\delta) - M_{\hat{\mathbf{F}}}\mathbf F\gamma +  P_{\hat{\mathbf{F}}}\mathbf A + P_{\hat{\mathbf{F}}}\mathbf E.\\
		\end{aligned}
	\end{equation*}
	Therefore, by the triangle inequality,
	\begin{equation*}
		\begin{aligned}
			\left\|\mathbf Q\hat\delta+\hat{\mathbf F}\hat\gamma - \mathbf Q\delta-\mathbf F\gamma
			\right\|_{NT} & \leq \frac{3.5r_{s,m}}{\kappa} + 2\|M_{\hat{\mathbf{F}}}\mathbf F\gamma\|_{NT} + \sqrt{2}\|\mathbf A\|_{NT} + \|P_{\hat{\mathbf{F}}}\mathbf{E}\|_{NT}.
		\end{aligned}
	\end{equation*}
\end{proof}

\subsection{Proof of Corollary~\ref{cor:fapd_rate}}

Let $\hat D=\mathrm{diag}(\sigma_1^2(\mathbf{X}),\ldots,\sigma_R^2(\mathbf{X}))$ be a diagonal matrix with diagonal entries equal to the squares of the largest $R$ singular values of $\mathbf{X}$.
\begin{lemma}\label{lem:factor_error_3}
	Under Assumption~\ref{as:fapd_factor} (i) and (ii), we have
	\begin{itemize}
		\item[(i)] $\|\mathbf{U}B\|_{\rm op}^2 = O_P({p_xNT})$;
		\item[(ii)] $\|\mathbf{U}\|_{\rm op}^2 = O_P\left(p_x(NT)^{1/2}+p_x^{1/2}NT\right)$;
		\item[(iii)] $\|\hat D^{-1}\|_{\rm op} = O_P\left(\frac{1}{p_xNT}\right)$.
	\end{itemize}
\end{lemma}
\begin{proof}[Proof of Lemma~\ref{lem:factor_error_3}]
	We prove each of the claims separately.

	\medskip
	\noindent\emph{Proof of (i).} Write $B=[b_1,\dots,b_R]$. Under Assumption~\ref{as:fapd_factor} (ii), we have
	\begin{equation*}
		\begin{aligned}
			\E\left[\|\mathbf{U}B\|_{\rm op}^2\mid B\right] & \leq \E\left[\|\mathbf{U}B\|_F^2\mid B\right] = \sum_{i,t,r}\E(u_{i,t}^\top b_r|B)^2 = \sum_{i,t,r}b_r^\top\E(u_{i,t}u_{i,t}^\top|B)b_r \leq CNT\sum_{r=1}^Rb_r^\top b_r.
		\end{aligned}
	\end{equation*}
	The result then follows under Assumption~\ref{as:fapd_factor} (i) and Markov's inequality.

	\medskip
	\noindent\emph{Proof of (ii).} The result follows by Markov's inequality and
	\begin{equation*}
		\begin{aligned}
			\E\|\mathbf{U}\|_{\rm op}^4 & = \E\|\mathbf{U}^\top\mathbf{U}\|_{\rm op}^2 \leq \E\|\mathbf{U}^\top\mathbf{U}\|_F^2 = \E\mathrm{tr}(\mathbf{U}\mathbf{U}^\top\mathbf{U}\mathbf{U}^\top) = \sum_{i,s}\sum_{j,t}\E(u_{i,s}^\top u_{j,t})^2 \\
			& = \sum_{i=1}^N\sum_{s=1}^T\E|u_{i,s}|_2^4 + \sum_{(i,s)\ne (j,t)}\E(u_{i,s}^\top u_{j,t})^2 = O(p_x^2NT) + O(p_xN^2T^2),
		\end{aligned}
	\end{equation*}
	where the last line follows by Assumption~\ref{as:fapd_factor} (vi).

	\medskip
	\noindent\emph{Proof of (iii).} By Weyl's inequality for singular values, see \cite{horn2013matrix}, Corollary 7.3.5,
	\begin{equation*}
		\begin{aligned}
			\left|\sigma_{R}(\mathbf{X}) - \sigma_R(\mathbf{F}B^\top)\right| & \leq \|\mathbf{U}\|_{\rm op} = o_P\left(\sqrt{p_xNT}\right),\qquad p_x,NT\to\infty,
		\end{aligned}
	\end{equation*}
	which follows by (ii). Therefore, under Assumption~\ref{as:fapd_factor} (i), we have
	\begin{equation*}
		\sigma_R(\mathbf{X}) \geq \sigma_R(\mathbf{F}B^\top) - o_P\left(\sqrt{p_xNT}\right) \geq \sigma_R(\mathbf{F})\sigma_R(B) - o_P\left(\sqrt{p_xNT}\right)\geq c\sqrt{p_xNT}
	\end{equation*}
	with probability approaching one. The result follows since $\|\hat D^{-1}\|_{\rm op} = \frac{1}{\sigma_R^2(\mathbf{X})}$.
\end{proof}

\begin{lemma}\label{lem:factor_error_2}
	Under the assumptions of Lemma~\ref{lem:factor_error_3}, we have
	\begin{equation*}
		\frac{1}{\sqrt{NT}}\|\hat{\mathbf{F}}-\mathbf{F}H^\top\|_{\rm op} = O_P\left(\frac{1}{\sqrt{p_x}} + {\frac{1}{\sqrt{NT}}}\right),
	\end{equation*}
	where $H^\top = B^\top B\mathbf{F}^\top\hat{\mathbf{F}}\hat D^{-1}$, $\hat D=\mathrm{diag}(\hat\sigma_1^2(\mathbf{X}),\ldots,\hat\sigma_R^2(\mathbf{X}))$.
\end{lemma}
\begin{proof}
	Using the factor model in Equation~\eqref{eq:factor_model}, expand
	\begin{equation*}
		\mathbf{X}\mathbf{X}^\top = \mathbf{F}B^\top B\mathbf{F}^\top + \mathbf{F}B^\top\mathbf{U}^\top + \mathbf{U}B\mathbf{F}^\top + \mathbf{U}\mathbf{U}^\top.
	\end{equation*}
	Since the columns of $\hat{\mathbf{F}}/\sqrt{NT}$ are the eigenvectors associated with $\hat D$, we have $\mathbf{X}\mathbf{X}^\top\hat{\mathbf{F}} = \hat{\mathbf{F}}\hat D$, and
	\begin{equation*}
		\hat{\mathbf{F}} = \mathbf{F}B^\top B\mathbf{F}^\top\hat{\mathbf{F}}\hat D^{-1} + \mathbf{F}B^\top\mathbf{U}^\top\hat{\mathbf{F}}\hat D^{-1} + \mathbf{U}B\mathbf{F}^\top\hat{\mathbf{F}}\hat D^{-1} + \mathbf{U}\mathbf{U}^\top\hat{\mathbf{F}}\hat D^{-1}.
	\end{equation*}
	or equivalently,
	\begin{equation*}
		\begin{aligned}
			\hat{\mathbf{F}} - \mathbf{F}H^\top & = \mathbf{F}B^\top\mathbf{U}^\top\hat{\mathbf{F}}\hat D^{-1} + \mathbf{U}B\mathbf{F}^\top\hat{\mathbf{F}}\hat D^{-1} + \mathbf{U}\mathbf{U}^\top\hat{\mathbf{F}}\hat D^{-1} \\
			& =: A_1 + A_2 + A_3.
		\end{aligned}
	\end{equation*}
	We bound each of the three terms separately. First, by Lemma~\ref{lem:factor_error_3} (i) and (iii), we have
	\begin{equation*}
		\begin{aligned}
			\|A_1\|_{\rm op} & \leq \|\mathbf{F}\|_{\rm op}\|B^\top\mathbf{U}^\top\|_{\rm op}\|\hat{\mathbf{F}}\|_{\rm op}\|\hat D^{-1}\|_{\rm op} \\
			& = O_P(\sqrt{RNT})O_P(\sqrt{p_xNTR})O_P(\sqrt{NT})O_P\left(\frac{1}{p_xNT}\right) \\
			& = O_P\left(R\sqrt{\frac{NT}{p_x}}\right),
		\end{aligned}
	\end{equation*}
	where the second line follows from $\|\hat{\mathbf{F}}\|_{\rm op}^2 = \|\hat{\mathbf F}^\top\hat{\mathbf{F}}\|_{\rm op} = NT$, and
	\begin{equation*}
		\|\mathbf{F}\|^2_{\rm op} \leq \|\mathbf{F}\|_{F}^2 = \operatorname{tr}(\mathbf{F}^\top\mathbf{F}) = O_P(RNT)
	\end{equation*}
	under Assumption~\ref{as:fapd_factor} (i).

	Second, by Lemma~\ref{lem:factor_error_3} (i) and (iii), we have
	\begin{equation*}
		\begin{aligned}
			\|A_2\|_{\rm op} & \leq \|\mathbf{U}B\|_{\rm op}\|\mathbf{F}^\top\|_{\rm op}\|\hat{\mathbf{F}}\|_{\rm op}\|\hat D^{-1}\|_{\rm op} \\
			& = O_P\left(\sqrt{p_xNTR}\right)O_P(\sqrt{RNT})O_P(\sqrt{NT})O_P\left(\frac{1}{p_xNT}\right) \\
			& = O_P\left(R\sqrt{\frac{NT}{p_x}}\right).
		\end{aligned}
	\end{equation*}
	Third, by Lemma~\ref{lem:factor_error_3} (ii) and (iii), we have
	\begin{equation*}
		\begin{aligned}
			\|A_3\|_{\rm op} & \leq \|\mathbf{U}\mathbf{U}^\top\|_{\rm op}\|\hat{\mathbf{F}}\|_{\rm op}\|\hat D^{-1}\|_{\rm op} \\
			& = O_P\left(p_x^{1/2}NT + p_x(NT)^{1/2}\right)O_P(\sqrt{NT})O_P\left(\frac{1}{p_xNT}\right) \\
			& = O_P\left(\sqrt{\frac{NT}{p_x}} + 1\right).
		\end{aligned}
	\end{equation*}
	Combining the bounds for $A_1$, $A_2$, and $A_3$ yields the desired result.
\end{proof}

\begin{lemma}\label{lem:factor_error}
	Under Assumptions~\ref{as:fapd_hl_score} and \ref{as:fapd_factor}, we have
	\begin{itemize}
		\item[(i)] $\|M_{\hat{\mathbf{F}}}\mathbf{F}\gamma\|_{NT}^2 = O_P\left(\frac{1}{p_x} + \frac{1}{NT}\right)$;
		\item[(ii)] $\|P_{\hat{\mathbf{F}}}\mathbf E\|_{NT}^2 = O_P\left(\frac{1}{NT}\right)$.
	\end{itemize}
\end{lemma}
\begin{proof}[Proof of Lemma~\ref{lem:factor_error}]
	We prove the two claims separately.

	\medskip
	\noindent\emph{Proof of (i).} Since $P_{{\mathbf{F}}}{\mathbf{F}}=\mathbf{F}$, we have
	\begin{equation*}
		\begin{aligned}
			\|M_{\hat{\mathbf{F}}}\mathbf{F}\gamma\|_{NT}^2 & = \|(P_{\mathbf{F}} - P_{\hat{\mathbf{F}}})\mathbf{F}\gamma\|_{NT}^2 \leq \|P_{\mathbf{F}} - P_{\hat{\mathbf{F}}}\|_{\rm op}^2\|\mathbf{F}\gamma\|_{NT}^2,
		\end{aligned}
	\end{equation*}
	where $\|\mathbf{F}\gamma\|_{NT}^2 = \frac{1}{NT}\gamma^\top\mathbf{F}^\top\mathbf{F}\gamma = O_P(1)$ under Assumption~\ref{as:fapd_factor} (i) and (iii). Next, by \citet{stewart1990matrix}, Theorem I.5.5,
	\begin{equation*}
		\begin{aligned}
			\|P_{\hat{\mathbf{F}}} - P_{\mathbf{F}}\|_{\rm op} & = \sqrt{2}\|P_{\hat{\mathbf{F}}}M_{\mathbf{F}}\|_{\rm op} = \frac{\sqrt{2}}{NT}\|\hat{\mathbf{F}}\hat{\mathbf{F}}^\top M_{\mathbf{F}}\|_{\rm op}
			\leq \frac{\sqrt{2}}{\sqrt{NT}}\|\hat{\mathbf{F}}^\top M_{\mathbf{F}}\|_{\rm op} \\
			& = \frac{\sqrt{2}}{\sqrt{NT}}\|(\hat{\mathbf{F}}^\top - H\mathbf{F}^\top)M_{\mathbf{F}}\|_{\rm op} \leq \frac{\sqrt{2}}{\sqrt{NT}}\|\hat{\mathbf{F}} - \mathbf{F}H^\top\|_{\rm op} \\
			& = O_P\left(\frac{R}{\sqrt{p_x}} + {\frac{1}{\sqrt{NT}}}\right),
		\end{aligned}
	\end{equation*}
	where the last line follows by Lemma~\ref{lem:factor_error_2}.

	\medskip
	\noindent\emph{Proof of (ii).} We now prove part (ii). Conditionally on $\mathbf{X}$, the matrix $P_{\hat{\mathbf{F}}}$ is a deterministic projection matrix.
	Under Assumption~\ref{as:fapd_factor} (iv), by the Hanson-Wright inequality, see \cite{Vershynin_2026}, Theorem 6.2.2, there exists $c>0$ such that for every $u>0$
	\begin{equation*}
		\Pr\left(\left|\mathbf E^\top P_{\hat{\mathbf{F}}}\mathbf E - \mathbb E[\mathbf E^\top P_{\hat{\mathbf{F}}}\mathbf E\mid \mathbf{X}]\right|\geq u\mid \mathbf{X} \right) \leq 2\exp\left(-c\min\left\{\frac{u^2}{\sigma^4\|P_{\hat{\mathbf{F}}}\|_F^2},\frac{u}{\sigma^2\|P_{\hat{\mathbf{F}}}\|_{\rm op}}\right\}\right),
	\end{equation*}
	where $\|P_{\hat{\mathbf{F}}}\|_F^2=R$ and $\|P_{\hat{\mathbf{F}}}\|_{\rm op}=1$ by the properties of projection matrices. Moreover,
	\begin{equation*}
		\mathbb E\left[\mathbf E^\top P_{\hat{\mathbf{F}}}\mathbf E\mid \mathbf{X}\right] = \operatorname{tr}\left(P_{\hat{\mathbf{F}}}\operatorname{Var}(\mathbf E\mid \mathbf{X})\right) \leq	C\operatorname{tr}(P_{\hat{\mathbf{F}}}) = CR,
	\end{equation*}
	where the last line follows under Assumption~\ref{as:fapd_factor} (v). This shows that
	\begin{equation*}
		\|P_{\hat{\mathbf{F}}}\mathbf E\|_{NT}^2 = \frac{1}{NT}\mathbf E^\top P_{\hat{\mathbf{F}}}\mathbf E = O_P\left(\frac{R}{NT}\right).
	\end{equation*}
\end{proof}

\begin{proof}[Proof of Corollary~\ref{cor:fapd_rate}]
	By Theorem~\ref{thm:fapd_oracle} and Assumptions~\ref{as:tuning}, we have
	\begin{equation*}
		\|\mathbf{Q}\hat\delta+\hat{\mathbf{F}}\hat\gamma-\mathbf{Q}\delta-\mathbf{F}\gamma\|_{NT}^2 = O_P\left(\frac{m\log(eK/m)+s\log(ed_{(m)}/s)}{NT} + \|{M}_{\hat{\mathbf{F}}}\mathbf{F}\gamma\|_{NT}^2 + \|P_{\hat{\mathbf{F}}}\mathbf E\|_{NT}^2 \right).
	\end{equation*}
	By Lemma~\ref{lem:factor_error}, $\|{M}_{\hat{\mathbf{F}}}\mathbf{F}\gamma\|_{NT}^2+\|P_{\hat{\mathbf{F}}}\mathbf E\|_{NT}^2 = O_P(p_x^{-1})$. The result follows since $p_x^{-1} = O(r_{s,m}^2)$.
\end{proof}

\bibliographystyle{plainnat}
\bibliography{references}

\end{document}